\documentclass[orivec]{llncs}
\pdfoutput=1

\usepackage[T1]{fontenc}
\usepackage{tabulary}
\usepackage{tabularx}
\usepackage{version}
\excludeversion{SHORT}
\includeversion{EXT}

\usepackage{listings}
\usepackage{hyperref}

\usepackage{amsthm}
\usepackage{amsfonts}
\usepackage{url}

\usepackage{rotating}
\usepackage{paralist}
\usepackage{amsmath}

\lstdefinelanguage{aslanpp}{
  basicstyle=\scriptsize\ttfamily, 
  breakatwhitespace=false,
  mathescape=true,
   numbers=left,xleftmargin=2em,framexleftmargin=1.5em,
   numberstyle=\tiny,
  morekeywords={specification,channel_model,CCM,ICM,ACM,entity,on,else,import,types,symbols,nonpublic,noninvertible,macros,clauses,equations,body,breakpoints,new,any,where,send,receive,over,retract,if,then,elseif,while,select,on,assert,constraints,goals,forall,exists,Actor,for} 
}
\newcommand{\fix}[2]{{\bf FIX}\footnote{{\bf #1:} #2 }}
\renewcommand{\fix}[2]{}

\newcommand{\sqli}{SQLi}
\newcommand{\bbb}{BB}
\newcommand{\tb}{TB}
\newcommand{\eb}{EB}
\newcommand{\uq}{UQ}
\newcommand{\sq}{SQ}
\newcommand{\so}{SO}

\newcommand{\sqlfast}{SQLfast}
\newcommand{\atse}{CL-AtSe}
\newcommand{\joo}{Joomla!}
\newcommand{\chained}{YAVWA}

\newcommand{\ifarrow}[1][]{
  \ifthenelse{\equal{#1}{}}{
    \Rightarrow
  }{
    =\!\!\![{#1}]\!\!\hspace{1pt}\!\Rightarrow
  }}

\newcommand{\Paragraph}[1]{\smallskip\emph{#1}}

\newcolumntype{s}{>{\hsize=1.2\hsize}X}
\newcolumntype{b}{>{\hsize=0.5\hsize}X}
\newcolumntype{Z}{>{\hsize=0.5\hsize\centering\arraybackslash}X}

\begin{document}
\begin{SHORT}
\title{\vspace*{-1cm}
Formal Analysis of Vulnerabilities of Web Applications Based on SQL Injection\thanks{This work was carried out while Marco Rocchetto was at the Universit\`a di Verona.}
\vspace*{-0.6cm}}
\end{SHORT}
\begin{EXT}
\title{
Formal Analysis of Vulnerabilities of Web Applications Based on SQL Injection\thanks{This work was carried out while Marco Rocchetto was at the Universit\`a di Verona.} \\ (Extended Version)
}
\end{EXT}

\begin{SHORT}
\author{Federico De Meo\inst{1} \and Marco Rocchetto\inst{2} \and Luca Vigan\`o\inst{3}
}
\institute{Dipartimento di Informatica, Universit\`a degli Studi di Verona, Italy \and
iTrust, Singapore University of Technology and Design, Singapore \and
Department of Informatics, King's College London, UK}
\end{SHORT}

\begin{EXT}
\author{Federico De Meo\inst{1} \and Marco Rocchetto\inst{2} \and Luca Vigan\`o\inst{3}
}
\institute{Dipartimento di Informatica, Universit\`a degli Studi di Verona, Italy \and
iTrust, Singapore University of Technology and Design, Singapore \and
Department of Informatics, King's College London, UK}
\end{EXT}



\maketitle
\begin{abstract}
We present a formal approach for the analysis of attacks that exploit SQLi
to violate security properties of web applications. We give a formal
representation of web applications and databases, and show that our
formalization effectively exploits SQLi attacks. We implemented our
approach in a prototype tool called SQLfast and we show its efficiency on
four real-world case studies, including the discovery of an attack on Joomla! that no other tool can find.
%
\end{abstract}

\section{Introduction}
\label{sec:intro}


\textbf{Motivations.\ } 
According to OWASP (the Open Web Applications Security
Pro\-ject~\cite{owasptop10}), \emph{SQL injection (\sqli{})} is the most
critical threat for the security of web applications (\emph{web apps}, for
short), and MITRE lists improper \sqli{} neutralization as the most
dangerous programming error~\cite{top25}. \sqli{} was first defined
in~\cite{sqliphrack} but, also due to the increasing complexity of web
apps, \sqli{}s can still be very difficult to detect, especially by manual
\emph{penetration testing} (\emph{pentesting}). 

A number of \sqli{} scanners have thus been developed to search for injection
points and payloads, most notably \emph{sqlmap}~\cite{sqlmap}, which allows
human pentesters to find \sqli{} vulnerabilities by testing
the web app with different payloads, and \emph{sqlninja}~\cite{sqlninja}, which
focuses on SQL server databases. The combination of the two provides the
pentester with a powerful tool suite for \sqli{} detection.
However, neither sqlmap nor sqlninjia (nor other state-of-the-art
vulnerability scanners) are able to detect vulnerabilities linked to
logical flaws of web apps~\cite{johnny}.
This means that even if a scanner can concretely discover a \sqli{}, it
can't link \sqli{} to lo\-gical flaws that lead to the violation of a
generic security property, e.g., the secrecy of data accessible only
bypassing an authentication phase via a \sqli{}.

Moreover, determining that 
a web app is vulnerable to
\sqli{} (and which 
payload to exploit) might not be enough for 
the app's overall security. Consider, for instance, a web app
that relies on legacy code (when an update is not feasible, e.g., 
because the legacy code is a core part of the system). If a
\sqli{} is found, an investigation should be performed to understand when
the \sqli{} can be exploited and whether this compromises security.
This investigation is carried out manually by the 
pentester in charge of 
identifying 
attack scenarios, thus potentially leading to additional omissions, errors and oversights in the security analysis.

A number of formal approaches for the security analysis of web apps, based
on the \emph{Dolev-Yao (DY) intruder model}~\cite{dolev1983}, have been
implemented recently, e.g.,
\cite{towards,avantssar-tacas,spacite,csrf,spacios}. However, the DY model
is typically used to reason about security protocols and the cryptographic
operators they employ (e.g., for asymmetric or symmetric cryptography,
modular exponentiation or exclusive-or) but abstracting away the contents
of the payloads of the messages. As a consequence, these approaches cannot
properly identify or exploit new \sqli{} payloads since reasoning about
the contents of the messages is crucial to that end.

\textbf{Contributions.\ } In this paper, we 
present a formal approach for the analysis of attacks that exploit \sqli{}
to violate security properties of web apps. We define how to formally
represent web apps that interact with a database and how the DY intruder
model can be extended to deal with \sqli{}.

In order to show that our formalization can effectively be used to detect
security vulnerabilities linked to \sqli{} attacks, we have developed a
prototype tool called \emph{\sqlfast{} (SQL Formal AnalysiS Tool)} and we
show its efficiency by discussing four real-world case studies. Most notably, we use \sqlfast{} to detect an attack on \joo{} which, to the
best of our knowledge, no state-of-the-art \sqli{} scanner (e.g., sqlmap
or sqlninja) can detect since they do not automatically link different
attacks in one attack trace (i.e., they do not find logical flaws linked
to \sqli{} attacks). Another key novel aspect of \sqlfast{} is that it can
detect complex attacks in which 
a first \sqli{} attack provides data for a second subsequent attack.
We show that \sqlfast{} allows us to exploit \sqli{} combining it with
logical flaws of web apps to report sophisticated attack traces in a few
seconds and can also deal with Second-Order \sqli{}s, which are notoriously
difficult to spot.



Note that we do not search for new \sqli{} payloads but rather we exploit
attacks related to \sqli{}. This allows us to analyze how an intruder can
violate a security property by exploiting one or more attacks related to a
\sqli{}, e.g., credential bypass. Nevertheless, we also (automatically)
test our attacks against the web app under analysis and then we use
state-of-the-art tools (i.e., sqlmap and curl) to detect the actual
payload of all the \sqli{}s exploited.


\textbf{Organization.\ }
In \autoref{sec:joomla}, we discuss a concrete example that shows why we can't stop at the identification of a \sqli{}. 
In \autoref{sec:sqli}, we give
a categorization of \sqli{} vulnerabilities\begin{EXT} to highlight the main aspects that characterize \sqli{}\end{EXT}, based on which, in \autoref{sec:formalization},
we provide our formalization. In \autoref{sec:casestudiesandresults}, we
discuss 
\sqlfast{} and 
its application
real-world case studies. 
\begin{SHORT}
In \autoref{sec:conclusions}, we sum up and discuss related and future work. 
The 
extended version~\cite{sqliextended} contains full details on our specifications and case studies, and a proof that the formalization of the database correctly handles all the SQLis categorized in \autoref{sec:sqli}.
\end{SHORT}
\begin{EXT}
We discuss related work in \autoref{sec:related}  and conclude in \autoref{sec:conclusions}. 
The appendix contains a proof that the formalization of the database correctly
handles all the \sqli{s} categorized in \autoref{sec:sqli}, as well as full
details on our specifications and case studies.
\end{EXT}


\section{Why can't we stop at the identification of a \sqli{}? The case of Joomla!}
\label{sec:joomla}
The identification of a \sqli{} entry point is generally considered as a
satisfactory finish line when dealing with \sqli{} in web apps. So, one
could ask: why not simply stop there (and why bother reading the rest of
this paper)? The answer is that a \sqli{} can be a serious threat
\emph{only} if it can be exploited and \emph{only} if it can be used for
carrying out an attack. A full understanding of how a potential \sqli{}
vulnerability can afflict the security of a web app is essential in order
to implement proper countermeasures. For instance, consider
\joo{}~\cite{Joomla}, a PHP-based Content Management System that allows
users to create web apps through a web interface. \joo{} supports
different databases, e.g., MySQL~\cite{MySQL} and
PostgreSQL~\cite{PostgreSQL}, and a recent assessment~\cite{joomlaattack}
has shown that versions ranging from 3.2 to 3.4.4 suffer from a \sqli{}
vulnerability~\cite{CVE-2015-7857}.

The execution of a state-of-the-art scanner such as sqlmap on \joo{} can
correctly find the vulnerability. However, sqlmap (or any other scanner for
\sqli{}) cannot tell how that \sqli{} can be usefully exploited in order
to carry out a concrete attack. A general description of the consequences
of a \sqli{} attack is given in~\cite{sqli:description,owasp:sqli} but, whenever a
\sqli{} entry point is found, the penetration tester has to manually
investigate the kind of damages that \sqli{} might cause to the web app.
The researchers who discovered the vulnerability of \joo{}~\cite{joomlaattack} 
also described how it could be exploited in a real attack: it
would allow an intruder to perform a session hijack and thus steal
someone's session but would not allow him to create his own account or
modify arbitrary data on the database. The exploration of different attack
scenarios has been entirely performed manually since no automatic tool
shows the outcome of the exploitation of a \sqli{} vulnerability on a
specific web app. But who guarantees that a post-\sqli{} attack can actually be performed and that all possible attacks based on the \sqli{} have been taken into account by the penetration tester?

This is why we can't stop at the identification of a \sqli{} and why we
can't address the post-\sqli{} attacks with a manual analysis. Our
approach addresses this by automating the identification of attacks that
exploit a \sqli{}.

\section{SQL injections}
\label{sec:sqli}
Some general classifications based on the payloads of the \sqli{}
(and the exploitation scenarios) have been put forth, e.g.,
\cite{halfond06mar,owasptop10}. Based on these, we can divide \sqli{}
techniques into 6 different categories: (i) Boolean-Based, (ii)
Time-Based, (iii) Error-Based, (iv) \texttt{UNION} Query, (v) Second-Order and (vi)
Stacked Queries.

Given that our formalization strictly depends on the attack that the
intruder wants to perform by using a particular type of \sqli{}, we now
define the two attacks that we have considered:\footnote{Other possible attacks (e.g., by exploiting a \emph{Cross-Site Scripting (XSS)} inside the payload of some \sqli{}) are outside the scope of our approach for now, cf.~\autoref{sec:conclusions}.}
\begin{compactitem}
\item \textit{Authentication bypass attack}: the intruder bypasses an
authentication check that a web app performs by querying a database.
\item \textit{Data extraction attack}: the intruder obtains data from the 
database that he should not be able to obtain.
\end{compactitem}
Based on these attack definitions, we will now describe the main details
of each category, emphasizing those aspects that are relevant for our
formalization. 
The following table summarizes which attacks can be exploited by a \sqli{} technique on a specific type of SQL query. 
\begin{EXT}Three remarks are in order: \end{EXT}
\begin{SHORT}Three remarks: \end{SHORT}
(1) since all state-of-the-art DBMS 
are vulnerable to \sqli{}, we won't distinguish between 
different dialects of SQL and simply write ``SQL query''; 
\begin{EXT}(2) for brevity, in the table we write \texttt{AB} for \emph{authentication bypass} and \texttt{DE} for \emph{data extraction};\end{EXT} 
\begin{SHORT}(2) \texttt{AB} abbreviates \emph{authentication bypass} and \texttt{DE} \emph{data extraction};\end{SHORT} 
(3) a scenario in which the intruder extracts information in order to bypass an authentication is considered to be a data extraction attack.

\begin{center}\scriptsize
\begin{tabularx} {\linewidth}[t] {s | Z  Z | Z Z | Z Z | Z Z | Z Z | Z Z} 
\multicolumn{1}{c}{} & \multicolumn{2}{c}{\bbb{}} & \multicolumn{2}{c}{\tb{}}  & \multicolumn{2}{c}{\eb{}}  &\multicolumn{2}{c}{\uq{}}  & \multicolumn{2}{c}{\so{}}  & \multicolumn{2}{c}{\sq{}} \\ \hline
                    & \texttt{AB}         & \texttt{DE}         & \texttt{AB}         & \texttt{DE}         & \texttt{AB}         & \texttt{DE}         & \texttt{AB}         & \texttt{DE}         & \texttt{AB}         & \texttt{DE}         & \texttt{AB}         & \texttt{DE}        \\  \hline
\texttt{SELECT}   & \checkmark & \checkmark & \checkmark & \checkmark & \checkmark & \checkmark &\checkmark & \checkmark &            &            &            & \checkmark \\ 
\texttt{UPDATE}   & \checkmark & \checkmark & \checkmark & \checkmark & \checkmark & \checkmark &           &            & \checkmark & \checkmark &            & \checkmark\\  
\texttt{DELETE}   &            &            &            &            &            & \checkmark &           &            &            &            &            & \checkmark\\  
\texttt{INSERT}   &            &            &            &            &            & \checkmark &           &            & \checkmark & \checkmark &            & \checkmark \\ \hline
\end{tabularx}
\end{center}

\begin{EXT}
\subsection{Boolean-Based \sqli{} (\bbb{})}
\label{sec:bbb}
\end{EXT}
\begin{SHORT}In a \textbf{Boolean-Based \sqli{} (\bbb{})},\end{SHORT}
\begin{EXT}In a \bbb{},\end{EXT}
an intruder inserts into an HTTP parameter, which is used by a web app
to write a SQL query,
one or more valid SQL statements that make the \texttt{WHERE}
clause of the SQL query evaluate to true or false. By interacting with the
web app and comparing the responses, the intruder can understand
whether or not the injection was successful. 
In this way, an intruder can perform both authentication bypass and
data extraction attacks.

In an authentication bypass attack, the intruder injects a statement that changes
the truth value of a \texttt{WHERE} clause in a SQL \texttt{SELECT} query, creating a
tautology. If a web app performs an authentication check querying a
database, this attack will then trick the database into replying in an
affirmative way even when no (or wrong) authentication details
have been presented by the intruder.

In a data extraction attack, the intruder obtains data from the database. The
term ``extraction'' is used in standard terminology but it can be misleading. 
With a \bbb{}, an intruder exploits the ``Boolean behavior'' of a
web app inferring whether the original query returned some tuples or
not. When the intruder understands how the web app behaves when some
tuples or 
no tuples are returned, he can start the ``extraction''. 
In this case, 
the intruder asks whether a certain information is
stored in the database and, based on the behavior of the web app, he
knows if the information is actually inside the database. 

\begin{EXT}
As an example,
consider a web app that presents a page \textit{p1} containing a
form. After the submission of the form, 
the web app
creates a query for the database searching for an entry that matches the
instantiation of the form fields submitted. 
Assume that the web app replies with \textit{p2} if a tuple is returned,
\textit{p1} otherwise. 
If the intruder injects a payload such as \texttt{or
username=admin} in one of the fields of the form and is redirected to
\textit{p2}, then he will know that \texttt{admin} is a valid value of the
database. 
\bigskip

\subsection{Time-Based \sqli{} (\tb{})}
\label{sec:tb}
\end{EXT}
\begin{SHORT}
A \textbf{Time-Based \sqli{} (TB)}\end{SHORT}
\begin{EXT}\tb{}\end{EXT}
is quite similar to \bbb{}:
the only difference is that \tb{} does not
need the web app to have a Boolean behavior. 
The intruder appends a timing function
to the validity value of a Boolean clause. 
Thus, after the submission of the query by the web app, the database waits for a predefined amount of time for a tuple as a response to the query;
the intruder can then infer whether the Boolean value of the
query was true or false observing a delay in the response.
\begin{EXT}
As an example, consider a web app that replies with a page
\textit{p1} independent of whether a tuple is returned by the database
(e.g., a search page). 
The intruder can 
inject a payload like \texttt{or
if(username=admin) wait 60s} and, if \texttt{admin} is a valid entry 
for the column \texttt{username}, the intruder will observe a delay of
around 60 seconds before receiving any answer. 
\end{EXT}
In real case scenarios, a \bbb{} is preferable 
as it is faster than a \tb{}. 
Timing is not part of our formalization (see \autoref{sec:formalization}), so
the abstract attack traces generated by our tool will not distinguish between
\bbb{} and \tb{}.

\begin{EXT}
\bigskip
\subsection{Error-Based \sqli{} (\eb{})}
\end{EXT}
When 
error pages are 
exposed to the Internet,
some error messages of the database could be exposed,
thus giving an intruder the possibility of exploiting an
\begin{SHORT} 
\textbf{Error-Based \sqli{} (EB)}.\end{SHORT}
\begin{EXT} 
\eb{}.\end{EXT}
In this type of injection, the intruder tricks the database
into performing operations that result in an error and then he extracts
information from the error messages produced by the database. 
EB is generally
used to perform a data extraction attack 
by inducing the generation of an error that contains some 
information stored in the database. 
\begin{EXT}
As an 
example, consider an intruder who wants to find out the
first username in the table \texttt{usernames}. He can inject, in
a login form, a payload that tricks the web app into evaluating the
query \texttt{SELECT * FROM (SELECT username FROM usernames LIMIT 1) AS tbl}. 
The web app generates 
\texttt{Error: table adminUsername unknown} because an
invalid table is selected (resulting from the inner query \texttt{SELECT
username FROM usernames LIMIT 1)}, where \texttt{adminUsername} is the first
username found in the user's table.
\bigskip

\subsection{\texttt{UNION} Query-Based \sqli{} (\uq)}
\label{sec:uq}
\end{EXT}

\begin{SHORT}
A \textbf{\texttt{UNION} Query-Based \sqli{} (\uq)}
\end{SHORT}
\begin{EXT}\uq{}\end{EXT}
is a technique 
in which an intruder injects a SQL \texttt{UNION}
operator to join the original query with a malicious one. 
The aim is to overwrite the values of the original query and thus, in
order to extract information, \uq{} requires the web app to print the
result of the query within the returned HTML page. This behavior allows
the intruder to actually extract information from the database by reading
it within the web app itself.
\begin{EXT}
As an 
example, consider the query \texttt{SELECT nickname FROM
users WHERE id=\$id}. An intruder can inject \texttt{\$1 UNION ALL SELECT
creditCardNumber FROM CreditCardTable} as id and thus obtain 
all the credit card numbers.
\bigskip

\subsection{Second-Order \sqli{} (\so{})}
\end{EXT}
\label{sec:secondOrder}

\begin{EXT}\so{}\end{EXT}
\begin{SHORT}
\textbf{Second-Order \sqli{} (\so{})}
\end{SHORT}
is an injection that 
has no direct effect when
submitted but that is exploited in a second stage of the attack. In some
cases, a web app may correctly handle and store a SQL statement whose
value depends on the user input. Afterwards, another part of the web
app that doesn't implement a control against \sqli{} might use the
previously stored SQL statement to execute a different query and thus expose
the web app to a \sqli{}. 
Automated 
scanners generally fail
to detect this type of \sqli{} (e.g., \cite{sqlmap,sqlninja}) and may need to
be manually instructed to check for evidence that an injection has been
attempted.
\begin{EXT}
As an example, consider a web app in which there is a user with
username \texttt{admin} (stored into the database). If this web app
has a registration page that allows special characters in the username, 
an intruder can register a new user with username \texttt{admin'\#} (where we
assume that \texttt{\#} is the comment delimiter character) and log in as
\texttt{admin'\#}. If the intruder changes the password of the user
\texttt{admin'\#}, 
a query like \texttt{UPDATE users SET password='123'
WHERE username='admin'\#'} is executed. 
The \emph{Database Management System (DBMS)} will interpret the \texttt{WHERE} clause as
\texttt{username='admin'} (because everything after the \texttt{\#} 
is considered a comment), so that the intruder succeeds in changing the
password of the admin user.
\bigskip

\subsection{Stacked Queries \sqli{} (\sq{})}

With a \sq{},\end{EXT}

\begin{SHORT}
With a \textbf{Stacked Queries \sqli{} (\sq{})},\end{SHORT}
an intruder can 
execute an arbitrary query different from the original one.
The semicolon character \texttt{;} enables the intruder to concatenate a
different SQL query to the original one.
By doing so, the intruder can perform data extraction attacks as well as
execute whatever operation is allowed by the database.
%
With a \sq{}, an
intruder can perform any of the \sqli{}s described above.
Thus, whenever we refer to all the \sqli{}s in our categorization, we
exclude \sq{} as it is already covered by the other ones.

\textbf{Prevention techniques.}
Avoiding \sqli{} attacks is theoretically quite straightforward. 
In fact, 
developers can use sanitization functions
or prepared statements.  
Roughly speaking, the general idea is to not evaluate the injected string as a SQL command. 

A \emph{sanitization function} takes the input provided by the user and removes (i.e., escapes) all the 
special characters that could be used to 
perform a \sqli{}. 
Sanitization functions are not the best option when dealing with 
\sqli{} because 
they might not be properly implemented or do not consider some cases.

\emph{Prepared statements} 
are the best option for preventing 
\sqli{}s. They are 
mainly used to execute the same query 
repeatedly maintaining 
efficiency. 
However, due to their inner execution principle
(if 
properly implemented) they are immune to \sqli{} attacks.
The execution of a prepared statement consists mainly in two steps:
preparation and execution. In the preparation step, the query is evaluated and 
compiled, waiting for the parameters for the instantiation. During the execution 
step, the parameters are submitted to the prepared statement and handled as data and 
thus they cannot be interpreted as SQL commands. 

\section{A Formalization of \sqli{}}
\label{sec:formalization}

We will now describe how we formally represent a web app that interacts
with a database using insecure SQL queries and/or a sanitized (i.e., secure) query.
In \autoref{sec:dy}, we propose an extension of the DY model that can deal with
\sqli{}.\footnote{This formal representation is intended to work with tools
		that perform symbolic analysis.
We don't formalize the honest client behavior and we assume the DY intruder to
be the only agent able to communicate with the web app. The DY intruder will
eventually perform honest interactions if needed to achieve a particular
configuration of the system. \begin{SHORT}See~\cite{sqliextended} for more
		details.\end{SHORT}}
We formalize the database in \autoref{sec:database}, the web app in
\autoref{sec:webapp}, and the goals in \autoref{sec:goals}. For brevity and readability, we omit many details and only give pseudo-code that should be quite intuitive. 
\begin{SHORT}See~\cite{sqliextended} for full details and the ASLan++ code of our formalizations and case studies, along with a brief introduction to ASLan++.\end{SHORT}
\begin{EXT}See the appendix for full details and the ASLan++ code of our formalizations and case studies, along with a brief introduction to ASLan++.\end{EXT}

\subsection{The DY web intruder}
\label{sec:dy}
We extend the standard DY intruder model~\cite{dolev1983} for 
security protocol analysis.
Suppose that we want to search for an 
authentication bypass attack via \bbb{}
(\autoref{sec:sqli}), in which
the intruder injects a statement that
changes the truth value of a \texttt{WHERE} clause in a SQL \texttt{SELECT} query, 
creating a tautology. 
To formalize this, we need to extend the DY intruder by giving him the
ability to send a concatenation of Boolean formulas made of conjunctions and
disjunctions. 
This characteristic highlights an
important difference between the classical DY intruder and the enhanced
version we are proposing: \emph{our web intruder works with abstract payloads
rather than messages}.
Due to technical details (e.g., implementation constraints and non-termination
problems), implementing such a modification is impractical.
We have thus 
allowed the intruder to
concatenate the exact payload, \textit{or true}, and defined a Horn clause to
model that whenever a formula has \textit{or true} injected by the intruder,
it 
evaluates to true.


We can rephrase the same reasoning in the case of \bbb{} for data extraction attacks, in which 
the intruder tricks the web app into asking to the database
if a particular information is present; 
for example,
instead of \texttt{or.true}, the intruder adds \texttt{or username=admin}. The
DBMS will reply in an affirmative way only if there is a tuple in the database
with \texttt{admin} as username.
To allow the intruder to perform all the \sqli{}s described in
\autoref{sec:sqli}, we thus extend the DY intruder with one constant
\texttt{sqli} that represents any \sqli{} payload (e.g., \texttt{or.true}). 


\subsection{The database}
\label{sec:database}
We give a general 
formalization of a database that can be used in any
specification to exploit \sqli{} when searching for security flaws in a web
app. 
Our formalization aims to be both \emph{compact}, 
to avoid state-space explosion problems, and \emph{general enough} not to be tailored to a given technology (e.g., MySQL or PostgreSQL). Hence, we don't represent the database content, the database structure, the SQL syntax nor access policies specified by the DBMS. Rather, we formalize messages sent and received and queries,
and a database can be seen as a network node that interacts only with the web
app through a secure channel.\footnote{Nothing prevents us from relaxing this
		assumption but this would give the DY intruder the possibility of
		performing attacks (e.g., man-in-the-middle attacks) that are rare in web
		app scenarios.} 
\begin{definition}
\emph{Messages} consist of \emph{variables} $V$, \emph{constants} $c$ ($\texttt{sqli}$,  
etc.), \emph{concatenation} $M.M$, \emph{function application} $f(M)$ of uninterpreted function symbols $f$ to messages $M$ (e.g., 
$\texttt{tuple}(M)$), and \emph{encryption} $\{M\}_M$ of messages with public, private or symmetric keys that are themselves messages.\begin{EXT}\footnote{In this paper we do not need to distinguish between the different kinds of encrypted messages, but we could of course do by following standard practice.} 
\end{EXT}
We define that $M_1$ is a \emph{submessage} of $M_2$ as is standard (e.g., $M_1$ is a submessage of $M_1.M_3$, of $f(M_1)$ and of $\{M_1\}_{M_4}$) and, abusing notation, write $M_1 \in M_2$.
\end{definition}

\begin{definition}
A query is \emph{valid} (respectively, \emph{not valid}) when, evaluated by a database, it returns one or more (respectively, zero) tuples.
%
\end{definition}
We formalize 
the validity 
of SQL queries by means of the Horn
clause:
\mbox{$\texttt{inDB(M.sqli)} \implies true$},
where, in order to represent a \sqli{} attack, the predicate \texttt{inDB()}
holds for a message (which represents a SQL query) whenever it is of the form
$M$.\texttt{sqli}. 
This states that the intruder has injected a payload \texttt{sqli} into the query parameters (expressed as a variable) $M$.

\Paragraph{Incoming messages.} 
We consider, as incoming messages, only SQL queries via raw SQL and via
sanitized queries. The parameters of queries are represented by a generic
variable \texttt{SQLquery}. In case of a raw SQL query, they are wrapped by an
uninterpreted function \texttt{query()}; if a sanitized query has been
implemented 
then we use another uninterpreted function
\texttt{sanitizedQuery()}. These two uninterpreted functions allow the modeler
to ``switch on/off'' 
the possibility of a \sqli{} in some point
of the app. 

\Paragraph{Database responses.} 
The tuple generated by the database as a response to a raw SQL query is
represented by an uninterpreted function \texttt{tuple()} over a message
representing a SQL query. Given that we do not model the content of the
database, this function represents any (and all) database data.

Whenever the database receives a SQL query \texttt{query(SQLquery)} from the
web app, the uninterpreted function \texttt{tuple(SQLquery)} is sent
back to the web app to express that a tuple, as a response to the
query, has been found. This response is returned only if
\texttt{inDB()} holds; in all other cases, a constant
\texttt{no\_tuple} is returned to represent that no tuples are returned in
the responses of the database. 

If the database receives a sanitized query, no injection is possible.
Hence, the database does not return any useful information to the
web app; instead, a constant \texttt{no\_tuple} is returned. Since 
the intruder 
cannot perform a \sqli{} in presence of a
sanitized query, we also assume that a sanitized query can be executed
only with legitimate parameters, i.e., as a function of \texttt{tuple()} (this
is because we are interested in modeling only \sqli{} scenarios).

The pseudo-code representing the database behavior is given in
\autoref{lst:databaseEntity}, where, here and in the following, we write
\texttt{DB} for the database.
\texttt{DB} is a network node and we assume it to be always actively
listening for incoming messages. 
It is defined by two main, mutually
exclusive, branches of an if-elseif statement: one guard is in line 1 in which
\texttt{DB} is waiting (expressed in Alice-and-Bob notation) for a sanitized
query and the other in line 3 in which it is waiting for a raw SQL query. If a
sanitized query is received, then there is no \sqli{}. Given that we only
consider dishonest interactions, the data sent back to the intruder will not
increase his knowledge. In other words, no \sqli{}s are permitted and any
permitted query will just give to the intruder the possibility of continuing
his execution with the web app but won't add any extra information to his
knowledge.
\lstset{numbers=left,xleftmargin=2em,framexleftmargin=1.5em}
\begin{lstlisting}[caption={Pseudo-code %representing
%of the behavior 
of a DBMS. %where %, for brevity, 
% (we write \texttt{DB} for the database).
},label={lst:databaseEntity}]
if(WebApp -> DB: sanitizedQuery(SQLquery)){
 if(SQLquery == tuple(*)) DB -> WebApp: no_tuple;
}elseif(WebApp -> DB: query(SQLquery)){
 if(inDB(SQLquery)) DB -> WebApp: tuple(SQLquery);
 if(!(inDB(SQLquery))) DB -> WebApp: no_tuple;}
\end{lstlisting}

One may argue that a valid query should indeed add extra information to the
intruder knowledge. However, we do not model the content of the database and
any information received by the intruder as a response to a sanitized query is
included in the action that the web app performs after this database response.
Thus, in our formalization, the query in \texttt{SQLquery} is not valid and
then the \texttt{no\_tuple} constant is sent back in line 2. We also add a
constraint in line 2 that any query received (\texttt{SQLquery}) must be of the
form \texttt{tuple(*)}, i.e., as a function of the content of the database
where \texttt{*} acts as a wildcard character that matches any possible
parameter. 
This is because, in the case of a sanitized query, the intruder cannot perform a \sqli{} and we 
exclude the case in which the DY intruder sends a random query just to continue the execution 
with the web app.
Instead, he has to either know a 
tuple of the database or data as functions of a tuple of the database. In the case the 
intruder knows \texttt{tuple(Query)}, he will just receive \texttt{no\_tuple}, i.e., 
correctly no data has been leaked 
to the intruder.

The second branch of the initial if-elseif statement (line 3) handles raw
queries. If a raw query is submitted, 
then there are two cases:
the raw query is not valid (line 5, where \texttt{!} formalizes the negation)
and then, as in the previous case, \texttt{no\_tuple} is sent back (line 5);
the raw query is valid (line 4) and a tuple is sent back (line 4). Given that
all these queries are sent from the intruder, we can assume they have a
malicious intent. One may argue that, in a real case scenario, the database is
not actually returning a tuple but, given that an intruder could repeatedly
send a \sqli{} exploiting that injection point, it is fair to assume that the
database is sending all the tuples it contains, i.e., \texttt{tuple(SQLquery)}.


\subsection{The web app}
\label{sec:webapp}
As for the database, the web app is a node of the network that can
send and receive messages. The web app communicates with a client or
with the database (it can, potentially, also communicate with other
apps but we do not consider that explicitly here). We assume only one
database is present because adding other databases would not add any further
useful information for finding attacks based on \sqli{}. The proof is
straightforward. 
Since we do not consider database contents and structures, if
we wanted to have two database models, then we would have two exact copies of
the formalization given in \autoref{sec:database}. Since we have assumed that
there exists a long-lasting secure relation between the database and the web
app, no man-in-the-middle attacks are considered. Therefore, any
attack found that involves the communication with one of the two databases
could be found by considering the other database only. 
%

A specification of a web app can be seen as
a behavioral description of the web app itself (along with its interaction
with the database). A modeler can define this specification from 
the design phase documentation of the
engineering process of the web app.
A model can also be created in a black box way by just looking at the HTTP
messages exchanged from a client and the web app and guessing the 
communication with the database.
%

We now consider the main aspects that allow for the modeling of 
a web app.

\Paragraph{Sending and receiving messages.}
A web app can communicate with a client and the database. 
We abstract away as many details as possible of the web pages and thus any incoming message will only contain:
(i) parameters of forms expressed as variables, e.g., \texttt{Client -> WebApp: Username.Password}, and
(ii) the web page itself expressed as a constant, e.g., \texttt{WebApp -> Client: dashboard} 
where \texttt{dashboard} 
represents a web page.
%
Note that, in any response of the web app, if the content of the response is linked to a response of the database, i.e., \texttt{tuple(Query)} (where the query is either \texttt{SELECT}, \texttt{UPDATE} or \texttt{DELETE}), then \texttt{tuple(Query)} must be included in the response. Otherwise, we would end up representing a scenario in which no content of the tuple received by the database is included in, or linked to, the web page 
and thus no \sqli{} would be present.

\Paragraph{Queries.} The web app creates either a sanitized query or a raw query.
Then, the web app wraps the variables representing the query parameters with either \texttt{sanitizedQuery()} or \texttt{query()}, both uninterpreted functions, and afterwards sends the SQL query to the database. Note that we only need to represent the parameters of a SQL query since we do not distinguish between different queries in the database formalization. 
If the query does not depend on parameters sent from a client, the intruder cannot exploit it to perform a SQLi.
The SQL query used to query the database is represented as a constant, resulting in the database always replying with \texttt{no\_tuple} (as \texttt{inDB()}, in this case, is never valid).

\Paragraph{If statements.}
We use them mainly to decide,
based on which kind of message has been received, what the web app
has to reply. For example, if the database replies with a tuple
\texttt{tuple(Query)}, then the web app might return a specific page along with
\texttt{tuple(Query)} or might return a different page.

\Paragraph{Assignments.}
A constant or a message can be assigned to a variable:
\texttt{Variable := constant|message}. 
Assignments are, e.g., useful to save incoming messages. 

\subsection{Goals}
\label{sec:goals}
Finally, we define the security properties we want the model to satisfy. As we discussed in \autoref{sec:sqli}, we consider two main attacks: authentication bypass and data extraction. We give the formalization in \autoref{lst:goals}, where \texttt{iknowledge} is a predicate that represents the knowledge of the intruder. 
By using the LTL ``globally'' operator \texttt{[]}, we
can specify an authentication check by stating that the intruder should not
have access to a specific page (\texttt{dashboard} in \autoref{lst:goals}), whereas data extraction is represented by
specifying that the intruder should not increase his knowledge with data from
the database (i.e., as function of \texttt{tuple()}).

\lstset{xleftmargin=0em,framexleftmargin=0em}
\begin{lstlisting}[numbers=none,caption={Authentication bypass and data extraction goals of the \bbb{} example.},label={lst:goals}]
[](!(iknowledge(dashboard))); %authentication bypass
[](!(iknowledge(tuple(*))));  %data extraction
\end{lstlisting}


\section{\sqlfast{}, Case Studies and Results}
\label{sec:casestudiesandresults}

To show that our formalization can be used effectively to detect security
flaws linked to 
\sqli{} attacks, 
we have developed \emph{\sqlfast{}}, a prototype 
\emph{SQL Formal AnalysiS Tool}~\cite{sqlfast}. 
In~\cite{sqlfast} we also provide a friendly web-based user interface
that helps the
modeler in 
creating the web-app model.
\sqlfast{}
\begin{SHORT}
		takes in input a specification written in AS\-Lan++, the modeling language of the AVANTSSAR Platform~\cite{avantssar-tacas},
and then calls \atse{} (one of the platform's model checkers) and generates
an \emph{Abstract Attack Trace (AAT)} as a \emph{Message Sequence Chart (MSC)}
if an attack was found. \sqlfast{} automatically detects which type of \sqli{}
was exploited and, in an interactive way, generates the curl or
sqlmap commands to concretize the attack.
\end{SHORT}
\begin{EXT}
		takes in input a specification written in
AS\-Lan++~\cite{aslan++}, the modeling language of the AVANTSSAR Platform
for security protocol analysis~\cite{avantssar-tacas}, and, by using the
ASLan++ translator, generates a transition system in the low-level
language ASLan~\cite{avantssar-deliverable-2.3update}. It then calls the
model checker \atse{}~\cite{atse} and generates an \emph{Abstract Attack
Trace (AAT)} as a MSC if an attack was found. \sqlfast{} automatically
detects which type of \sqli{} has been exploited and, in an interactive
way, generates the curl or sqlmap~\cite{sqlmap} commands to concretize the
attack 
reported.
\end{EXT}

%


As a concrete proof-of-concept, we have applied \sqlfast{} to (i) WebGoat~\cite{WebGoat}, (ii) Damn Vulnerable Web Application (DVWA)~\cite{dvwa}, (iii) \joo{} 3.4.4, and (iv) \emph{Yet Another Vulnerable Web Application (\chained{})}, an ad-hoc testing environment that we have developed and that also includes
a \so{} \sqli{} example.
\begin{SHORT}(Recall that full details are given in~\cite{sqliextended}.)\end{SHORT} The case studies provided by WebGoat and DVWA might sound
limited but 
capture all possible scenarios with respect to \sqli{}
attack combinations considered in this paper --- 
recall that our formalization for \sqli{} attacks does
not find \sqli{} payloads, but focuses on vulnerabilities based on \sqli{}. We tested SQLfast in order to show all the combinations that could be represented by considering \sqli{} for (1) authentication bypass, (2) data extraction and (3) data extraction with reuse of the extracted information. Our case studies are quite heterogenous, so it should not be difficult to map other case studies to one of these scenarios we have considered.

We have implemented the case studies in ASLan++ to be able to
apply the model checkers of the AVANTSSAR Platform (in particular, CL-AtSe), but other model checkers implementing the Dolev-Yao intruder model could be used as well, provided that their input language is expressive enough. For the sake of brevity, we discuss here only the case studies \joo{}, \chained{} and \so{}, which show how our formalization can find attacks linked to the logic of a web app that is vulnerable to \sqli{} attacks. The type of 
attacks that \sqlfast{} can detect and concretize cannot be 
detected by 
state-of-the-art tools for \sqli{} such as sqlmap.  

\subsection{Case Study: Authentication Bypass via Data Extraction}
\label{sec:authdata}
We now discuss two scenarios in which our approach detects attacks that
state-of-the-art tools, such as sqlmap, are \emph{not} able to detect and
exploit.  In the first scenario, the intruder exploits a recent \sqli{}
vulnerability found (by manual inspection only) in \joo{}~\cite{CVE-2015-7857}.
The second scenario (\chained{}) is a variant of the first 
and shows a concatenation of different attacks.

\Paragraph{\joo{}.}
A recent assessment has shown that the Content History module of \joo{} suffers from a \sqli{} vulnerability that allows a remote (non-authenticated) user to execute arbitrary SQL commands~\cite{CVE-2015-7857}.
The pseudo-code in~\autoref{lst:joomlamodel} represents the following
behavior: a remote user visits the Content History
component (line 1). The web app queries the database with the user
supplied data (2). If some tuples are generated (3), the web app sends
to the client the history page \texttt{viewHistory} along with the
\texttt{tuple()} function (4). The web app then has two possible ways of
authenticating the user (5--9): by using credentials or cookies. If
username and password are provided (5), the web app applies a
non-invertible hash function \texttt{hash()} to the password, and queries
the database to verify the credentials (6).\footnote{The web app applies a
hash function to the password before checking whether credentials are correct
because \joo{} stores the passwords hashed into the database.} If the
credentials are correct, the administration panel is sent to the user
(7). In case of a cookie session, the user provides a cookie that the
web app checks querying the database (8). If the cookie is valid,
the administration panel is sent back to the user (9).
%
%
\begin{lstlisting}[numbers=left,xleftmargin=2em,framexleftmargin=1.5em,breaklines=true,caption={Pseudo-code representing the \joo{} scenario.},
label={lst:joomlamodel}]
User -> WebApp: com_contenthistory.history.Listselect;
WebApp->DB: query(com_contenthistory.history.Listselect);
if(DB -> WebApp: tuple(SQLquery)){
 WebApp -> User: viewHistory.tuple(SQLquery); }
if(User -> WebApp: Username.Password){
 WebApp -> DB: sanitizedQuery(Username.hash(Password));
 if(DB -> WebApp: no_tuple){ WebApp -> User: adminPanel; }}
if(User -> WebApp: Cookie){ WebApp -> DB: sanitizedQuery(Cookie);
 if(DB -> WebApp: no_tuple){ WebApp -> User: adminPanel; }}
\end{lstlisting}
As goal, we check if there exists an execution
in which the intruder can access the administration panel represented by
the constant \texttt{adminPanel}.
\begin{lstlisting}[numbers=none,breaklines=true,caption={Authentication bypass for the \joo{} scenario.}
,label={lst:joomlagoal}]
[](!(iknowledge(adminPanel)));
\end{lstlisting}
%
\sqlfast{} generates the AAT in \autoref{lst:joomlaaat}, which 
is an
authentication bypass at\-tack where the intruder hijacks a user session by using a cookie instead of login credentials.  In fact, the web app applies a hash function to the password before verifying the credentials submitted by the user. The hash function would not allow an intruder to blindly submit a password extracted from the database, the only possibility is using a valid cookie value.\footnote{We do not consider the possibility of brute forcing the hashed password, in accordance with the perfect cryptography assumption of the DY model.}
The intruder 
performs a data extraction 
and retrieves the information 
to access the administration panel (1--4), and 
uses it 
to hijack a user session by submitting a valid cookie value (5--8).

\begin{lstlisting}[numbers=left,xleftmargin=2em,framexleftmargin=1.5em,breaklines=true,caption={Abstract attack trace that extracts data with a \sqli{} in order to bypass the authentication of the \joo{} scenario.}
,label={lst:joomlaaat}]
i -> WebApp  : com_contenthistory.history.sqli
WebApp -> DB : query(com_contenthistory.history.sqli)
DB -> WebApp : tuple(com_contenthistory.history.sqli)
WebApp  -> i : viewHistory.tuple(com_contenthistory.history.sqli)
i ->WebApp   : cookie.tuple(com_contenthistory.history.sqli)
WebApp -> DB : sanitizedQuery(tuple(com_contenthistory.history.sqli))
DB -> WebApp : no_tuple
WebApp  -> i : adminPanel
\end{lstlisting}

\Paragraph{\chained{}.}
\label{sec:chained}
We have designed a variant of \joo{} to show that
a \sqli{} can be exploited to compromise a part of a web app that does not 
directly depend on databases.
\chained{} 
provides 
an
HTTP form login and a login by HTTP basic 
authentication~\cite{rfc:basicauth} configured with the 
\texttt{.htaccess}~\cite{apache:htaccess} file. The credentials used for the
HTTP basic authentication, which are stored in the \texttt{.htpasswd} file, are
the same as the ones employed by the users to
login into the web app (i.e., the same as the ones stored in the
database).\begin{EXT}\footnote{One might see this as a bad practice and thus consider this whole example to be unrealistic. However, this scenario can easily be categorized as a security misconfiguration ($5^\textit{th}$ most critical security issue according 
to~\cite{owasptop10}). Moreover, recent	events~\cite{thehackernews} have shown, once again, how humans tend to reuse passwords across multiple logins.}\end{EXT}
The intruder's goal is to access the area protected by the HTTP 
basic authentication login. Obviously, 
he cannot perform a \sqli{} to
bypass HTTP basic authentication since 
the login procedure
doesn't use SQL. 
Bypassing
the login page, without knowing the correct credentials, 
doesn't allow
the intruder to gain access to the secure folder. 

We have defined this scenario in the pseudo-code in
\autoref{lst:chained}. The client sends his personal credentials
(\texttt{Username.Password}) to the web app (1). The web app
creates a query that it sends to the database (2) for verifying the
submitted credentials. If tuples are generated from the database (3), a
dashboard page is returned to the client along with the function 
\texttt{tuple()} (4), otherwise,
the web app redirects the user to
the login page (5).
At this point, the web app waits to receive correct credentials that will
allow the client to access the secure folder \texttt{secureFolder} (6).
Given that the credentials are the same as the ones stored in the
database, and the database content is represented with the function
\texttt{tuple()}, we can also represent credentials here with the function
\texttt{tuple()}.\footnote{We recall from \autoref{sec:database} that \texttt{tuple()}
represents an abstraction of any data that can be extracted from the
database. This means that whenever a web app requires any data 
in
the domain of the database, we can write them as a function of \texttt{tuple()}.}

\begin{lstlisting}[numbers=left,xleftmargin=2em,framexleftmargin=1.5em,breaklines=true,caption={Pseudo-code representing the \chained{} scenario.},
label={lst:chained}]
User -> WebApp: Username.Password;
WebApp -> DB: query(Username.Password);
if(DB -> WebApp: tuple(SQLquery)){
 WebApp -> User: dashboard.tuple(SQLquery);
}elseif(DB -> WebApp: no_tuple){ WebApp -> User: login; }
if(User -> WebApp: tuple(*)){ WebApp -> User: secureFolder; }
\end{lstlisting}
As goal, we check if the intruder can reach \texttt{secureFolder}.
%
\sqlfast{} generates the AAT given in \autoref{lst:aatChained}, in which 
the intruder 
successfully retrieves information from the database 
and 
uses such information to access a protected folder. 
The intruder performs a data extraction attack using \sqli{} 
(1--4), which allows him to retrieve information stored in the database, 
and then (5--6) submits the extracted data 
and accesses the restricted folder 
\texttt{secureFolder}.

\begin{lstlisting}[numbers=left,xleftmargin=2em,framexleftmargin=1.5em,breaklines=true,caption={Abstract attack trace
of the \chained{} case study.},
label={lst:aatChained}]
User -> WebApp: Username(4).sqli
WebApp -> DB  : query(Username(4).sqli)
DB -> WebApp  : tuple(Username(4).sqli)
WebApp -> i   : dashboard.tuple(Username(4).sqli)
i  -> WebApp  : tuple(Username(4).sqli)
WebApp  -> i  : secureFolder
\end{lstlisting}
\begin{SHORT}\vspace*{-0.5cm}\end{SHORT}

\subsection{Case Study: Second-Order \sqli{} (SO)}
\label{sec:secondorder}
We now show that our formalization is flexible enough to represent 
\so{}s, which are notoriously very difficult to detect and exploit.

This scenario is part of \chained{} and implements a web app that allows
users to register a new account. In the registration process, the web app
executes an (\texttt{INSERT}) SQL query that stores the user's credentials
into a database. The intruder can create an account submitting malicious
credentials that don't result in a \sqli{} but will trigger an injection
later on in the web app. After the registration phase, the user submits a
request for accessing an internal page. The web app performs another SQL
query using the same parameters previously used in the registration
process (i.e., the registration credentials). At this point, a page is
showed together with the injection and the intruder can exploit a \so{}.

We have formalized this scenario in~\autoref{lst:secondOrder}: a client
sends a registration request along with his personal credentials
(\texttt{Username} and \texttt{Password}) to the web app (1). The web app sends
a query containing the client's credentials to the database (2). The web app
checks if it receives a response from the database containing the data
resulting from the execution of the query \texttt{tuple(SQLquery)} submitted by
the web app (3). 
The web app sends back to the client the page \texttt{registered} (4). Here, the web app
does not forward \texttt{tuple()} 
because the registration
query is an \texttt{INSERT} (see \autoref{sec:database}).  The client asks for
a page (5), which makes the web app use previously submitted values of
\texttt{Username} and \texttt{Password} to execute a new SQL query (6). Here is
where the \so{} takes place; 
the variables embedded in the query in
(6) will trigger a \so{}.  The database executes the query and sends back the
results to the web app (7). Finally (8), the web app sends (by using a \texttt{SELECT} query) to the client the
requested page and the 
\texttt{tuple()}.\footnote{Recall that we don't represent SQL syntax in our
		models, so we don't explicitly represent the type of the SQL according to
		the modeling guidelines~in~\autoref{sec:webapp}.}

\begin{lstlisting}[numbers=left,xleftmargin=2em,framexleftmargin=1.5em,breaklines=true,caption={Pseudo-code representing
a web app vulnerable to a \so{} attack.},label={lst:secondOrder}]
User -> WebApp: registrationRequest.Username.Password;
WebApp -> DB: query(Username.Password);
if(DB -> WebApp: tuple(SQLquery)){
 WebApp -> User: registered;
 User -> WebApp: requestPage;
 WebApp -> DB: query(Username.Password);
 DB -> WebApp: tuple(SQLquery);
 WebApp -> User: page.tuple(SQLquery); }
\end{lstlisting}
As goal, we ask if 
the intruder can interact with the web app until he obtains data from the
database, i.e., with a data extraction attack, as in~\autoref{lst:goals}.
%
\sqlfast{} generates the AAT 
in \autoref{lst:aatSecondOrder}, in which
the intruder performs the registration process (1--4) by registering
malicious credentials \texttt{Username(4)} and \texttt{sqli}. At the end of the
registration process (5), the intruder asks for \texttt{requestPage} that makes
the web app send to the database a SQL query with the same parameters the
intruder used in the registration (6--7). In 
(8), the
intruder receives the requested page and the result of the execution of the
injected SQL query performing a \so{}.
\begin{lstlisting}[numbers=left,xleftmargin=2em,framexleftmargin=1.5em,breaklines=true,caption={Abstract attack trace 
for the \so{} case study.},label={lst:aatSecondOrder}]
User -> WebApp: registrationRequest.Username(4).sqli
WebApp -> DB  : query(Username(4).sqli)
DB -> WebApp  : tuple(Username(4).sqli)
WebApp -> i   : registered
i -> WebApp   : requestPage
WebApp -> DB  : query(Username(4).sqli)
DB -> WebApp  : tuple(Username(4).sqli)
WebApp -> i   : page.tuple(Username(4).sqli)
\end{lstlisting}

\begin{SHORT}\vspace*{-0.3cm}\end{SHORT}
\subsection{Concretization phase}
We executed \sqlfast{} on all our case studies 
using a standard laptop (Intel i7 with 8G RAM). The
execution time of the model-checking phase of \sqlfast{} 
ranges from 35 to 45 ms. The overall process (from translation to
concretization) 
takes a few seconds.  In all the cases, we 
generated AATs violating the security property we defined over the model
(authentication bypass or data extraction attack).  Once the AAT has been
generated, \sqlfast{} interactively asks 
the user to provide information
such as the URL of the web app. 
Finally, if
we are concretizing a \sqli{} that exploits an authentication bypass attack a curl
command is showed, whereas sqlmap is used for data extraction \sqli{}.  By
executing the 
traces generated by \sqlfast{}, we 
exploited all the
AATs over the real web app.

\begin{EXT}
\section{Related work}
\label{sec:related}
\end{EXT}
\begin{SHORT}
\section{Conclusions, related work, and future work}
\label{sec:conclusions}

We have presented a formal approach for the representation of \sqli{} and
attacks that exploit \sqli{} in order to violate security properties of
web apps. We have formally defined web apps that interact with a database (that properly replies to queries containing \sqli{}) and an extended DY intruder able to deal with authentication bypass and data extraction attacks related to \sqli{}. We have shown the efficiency of our prototype tool SQLfast
on four real-world case studies (see also~\cite{sqliextended}). \sqlfast{} handles \so{} and detects multi-stage attacks and logical flaws that, to the best of our knowledge, no other tool can handle together, and hardly ever even individually, including the discovery of an attack on Joomla!.
\end{SHORT}	

Many works have proposed
new \sqli{}
techniques and payloads (e.g.,
\cite{Stampar2013,OsandaMalithJayathissa,Damele2009}) 
or formal approaches 
to detect \sqli{} 
(e.g., \cite{Martin2008,Kiezun2009,amnesia,sqli:java}). 
However, to the best of our knowledge, ours is
the first attempt to 
search for 
vulnerabilities based on \sqli{} rather than to 
detect \sqli{}.  
There are, however, a number of works that are closely related to ours and
that are thus worth discussing.

 
\emph{SPaCiTE} is a model-based security testing tool for web apps that
relies on mutation testing~\cite{spacite}. SPaCiTE starts from a secure
ASLan++ specification of a web app and automatically introduces flaws by
mutating the specification. The strength of this approach is the
concretization phase. Starting from an AAT, generated from the mutated
specification using a model-checking phase, SPaCiTE concretizes and tests
the attack trace on the real web app. The major differences with respect to our approach reside in how we model web apps and in particular
those aspects that strictly characterize \sqli{} aspects. The
main goal of the approach in~\cite{spacite} is to find \sqli{} entry
points and concretize them, 
our main goal is to consider \sqli{}
aspects that can be exploited to attack a web app.

Another formal approach that uses ASLan++ and the DY intruder model for
the security analysis of web apps is~\cite{csrf}. In this work, the
authors model a web app searching for CSRF and they do not consider
databases or extensions to the DY model. However, the idea and the
representation of web apps is close to ours and we envision some
potentially useful interaction between the two approaches.

In~\cite{CalviVigano16}, the authors describe the ``Chained Attack''
approach, which considers multiple attacks 
to compromise a web app. The idea is close to ours, but: (i) they consider
a new kind of web intruder, whereas we stick with the DY intruder; (ii)
we analyzed the most common \sqli{} techniques and proposed a
formalization of a vulnerable database, they only consider the
behavior of the web app.

In~\cite{towards}, the authors 
present a model-based method for the
security verification of web apps. They propose a methodology for modeling
web apps and model 5 case studies in Alloy~\cite{alloy}.
Even if the idea is similar to our approach, they have defined three
different intruder models that should find web attacks, whereas we have
used (and extended) the standard DY one. Their AATs are difficult to
interpret because no MSCs are given but state configurations. They have
also considered a number of HTTP details that we have instead 
abstracted away 
in favor of an easier modeling phase.
In contrast, we display AAT as MSCs and we
proposed a concretization phase to obtain the concrete payloads of
\sqli{}.

\begin{EXT}
\section{Conclusions and future work}
\label{sec:conclusions}

We have presented a formal approach for the representation of \sqli{} and
attacks that exploit \sqli{} in order to violate security properties of
web apps. We have formally defined web apps that interact with a database (that
properly replies to queries containing \sqli{}) and an extended DY intruder
able to deal with authentication bypass and data extraction attacks related to
\sqli{}. We have shown the efficiency of our prototype tool SQLfast
on four real-world case studies \begin{SHORT}(see
		also~\cite{sqliextended})\end{SHORT}. \sqlfast{} handles \so{} and detects
multi-stage attacks and logical flaws that, to the best of our knowledge, no
other tool can handle together, and hardly ever even individually, including the
discovery of an attack on Joomla!.
\end{EXT}

%


As future work, we plan to extend the database formalization in order 
to consider \sqli{} that would modify the database state leading to more
complex \sqli{} exploitations. We also plan to analyze other web app vulnerabilities such as
stored/reflected XSS and broken session management, and
investigate synergies between our approach and the one of~\cite{csrf} on CSRF. We will extend our approach to detect
(i) complex concatenations of vulnerabilities (similar to, and more complex than, \cite{homakov}) that lead
to concatenations of attacks, and (ii) articulated paths to vulnerabilities
that would hardly ever be discovered by manual
analysis.

\bibliographystyle{abbrv}
\bibliography{literature}

\begin{EXT}

\appendix

\section{Correctness of the database formalization}
\label{sec:dyproof}
In this section, we prove that 
the
formalization of the database correctly handles all the \sqli{}s categorized
in \autoref{sec:sqli}.
We remark that tool, specification and 
guidelines show that it is concretely possible to exploit these SQLis.

\begin{theorem}
\label{thm:dysqli}
Let $i$ represent our extension of the DY intruder (cf.~\autoref{sec:dy}) and  
$\mathit{db}$ represent the database. For any message $Q$ such that \texttt{sqli} is a submessage of $Q$ (i.e., $\texttt{sqli} \in Q$), we have:
\begin{center}
\begin{tabular}{ll}
if & $i\rightarrow \mathit{db}: \texttt{sanitizedQuery(}\mathit{Q}\texttt{)}$ or \\
& $i\rightarrow \mathit{db}: \texttt{query(}\mathit{Q}\texttt{)}$\\
then & $\mathit{db} \rightarrow i: \mathit{Response}$
\end{tabular}
\end{center}
where $i \rightarrow db: f(M)$ means that $i$ sends to $\mathit{db}$ a message $M$ through 
the web app, which creates a query from the message $M$ either using a raw query if 
$f=\texttt{query}$, and a sanitized query otherwise, and $\mathit{Response}$ is (according 
to the formalization proposed in \autoref{sec:database}) a variable such that its 
instantiation is either
\texttt{no\_tuple} if \texttt{sanitizedQuery} has been used, or
\texttt{tuple(}$\mathit{Q}$\texttt{)} otherwise, with \texttt{tuple(}$\mathit{Q}$\texttt{)} representing 
the exploitation of a 
\sqli{} $s\in Q$.
\end{theorem}

\begin{proof}
\label{proof:dysqli}
We have extended the DY intruder with \texttt{sqli} to represent the payload that the intruder can use to send a \sqli{} to the web app and then to the database.

Given that $\texttt{sqli} \in Q$,  we only have two subcases: the query $Q$ is  either (i) 
a sanitized query or (ii) a raw query.
\begin{itemize}
\item[(i)] Line 1 of the database (whenever in the proof we refer to a line of the database we are referring to \autoref{lst:databaseEntity}) accepts the incoming message and returns the constant \texttt{no\_tuple} if $Q$ 
contains \emph{only} a tuple of the database (line 2). In this case, \texttt{SQLquery!=tuple(*)} because contains \texttt{sqli} 
and nothing is returned (i.e., a \sqli{} is detected). Correctly, no \sqli{} can be performed in this case.
\item[(ii)] Line 3 of the database accepts the incoming message and, given
that the message contains a \sqli{}, the predicate \texttt{inDB(Q.sqli)} is
true (due to the Horn clause $\texttt{inDB(M.sqli)} \implies true$, defined in \autoref{sec:database}), and 
the guard in line 4 
is also true. The database then replies with \texttt{tuple(Q.sqli)} and a \sqli{} 
has been exploited.
\end{itemize}

\end{proof}


\section{ASLan++ formalization}
\label{sec:aslan}
Although in the body of the paper we have shown only intuitive pseudo-code, we 
have actually implemented our formalizations using ASLan++. 
In \autoref{sec:aslanintro}, we give a brief overview of ASLan++~\cite{aslan++} 
and  ASLan~\cite{avantssar-deliverable-2.3update}, which are used by the 
AVANTSSAR Platform~\cite{avantssar-tacas} and the SPaCIoS Tool~\cite{spacios} 
for the formal and automated analysis at design-time and validation at run-time 
of security-sensitive protocols and web apps and their security
goals.
We focus on the aspects that we use for modeling our case studies of
\autoref{sec:casestudiesandresults}. In \autoref{sec:skeleton}, we describe an
ASLan++ specification skeleton that encapsulates all the aspects considered in
\autoref{sec:formalization} and then expand this to full specifications of the
case studies in \autoref{sec:aslan-casestudies}.


\subsection{ASLan++ and ASLan}
\label{sec:aslanintro}

ASLan++ is a formal and typed security protocol specification language, whose
semantics is defined in terms of ASLan, which we describe below. Similar to
object-oriented languages, an ASLan++ specification consists in a hierarchy of
\emph{entity declarations}, which are similar to Java classes.

The top-level entity is usually called \emph{Environment} (similar to the
``main'' procedure of a program) and it typically contains the definition of a
\emph{Session} entity, which in turn contains a number of subentities (and
their instantiations, i.e., \texttt{new subentity(<parameters>);}) that
define the main principals involved in the system (e.g., web app, database).
Each subentity defines the internal behavior of the component it models and
the interaction with other entities. We won't go into the full details of the
ASLan++ language~\cite{avantssar-deliverable-2.3update}, but in the following
we discuss in itemized form the main aspects that we use to give the ASLan++
code of our formalizations.
\begin{itemize}
\item ASLan++ supports variables (capital letter symbols), constants (lower
case symbols), functions and predicates that can be distinguished by the
return type, message and fact (described afterwards) respectively.

\item An entity is composed by two main sections: 
\begin{itemize}
\item \texttt{symbols}, in which types of local symbols are declared and
instantiations are given of all the variables and constants used in the
entity, and
\item \texttt{body}, where the behavior of the entity is defined. The instantiation of an entity must be done in the body of the parent-entity using the \texttt{new} keyword as follows: \texttt{new subentity(<params>)}.
\end{itemize}
Inside the body of an entity we use three different types of statements:
assignments, message send and message receive.
The \emph{assignment} \texttt{Var:=fresh()} assigns to the variable \texttt{Var} a new constant of the proper type, while \texttt{Var:=Var'} is the usual assignment. 

\item Sending and receiving of messages are expressed in Alice-and-Bob
notation: \texttt{A -> B: M;}, where \texttt{A} and \texttt{B} are entities
and \texttt{M} a message. More formally, a \emph{message send} statement,
\texttt{Snd -> Rcv: M}, is composed by two variables \texttt{Snd} and
\texttt{Rcv} representing sender and receiver, respectively, and a message
\texttt{M} exchanged between the two parties. In \emph{message receive},
\texttt{Snd} and \texttt{Rcv} are swapped and usually, in order to assign a
value to the variable \texttt{M}, a \texttt{?} precedes the message
\texttt{M}, i.e., \texttt{Snd -> Rcv: ?M}. However, in ASLan++, the
\texttt{Actor} keyword refers to the entity itself (similar to ``this'' or
``self'' in object-oriented languages) and thus we actually write the send and
receive statements as \texttt{Actor -> Rcv: M} and \texttt{Snd -> Actor: ?M}
respectively. Note that \texttt{?} acts as a wildcard if it is not followed by any variable
(e.g., \texttt{Snd -> Actor: ?}) since no specific pattern of the receiving 
message is expected.

\item It is possible to use different kinds of channels: \texttt{->},
\texttt{*->}, \texttt{->*} and \texttt{*->*}, which define insecure,
authentic, confidential and secure channels, respectively. In ASLan++, there
are different types of channel models but we only use the Cryptographic
Channel Model (CCM).

\item A section \texttt{clauses} defines Horn clauses of the form:
\texttt{HCname: head :- body}.

\item In the \texttt{goals} section, one can specify LTL goals like \texttt{[]predicate(Var)} stating that a particular predicate must always hold over a variable \texttt{Var}.

\item The two statements \texttt{if(<guard>)\{positive branch\} else
\{negative branch\}} and \texttt{select\{on:\{<guard>\}\} \{positive branch\}}
are equivalent but select-on does not provide the negative branch.

\item \texttt{while(<guard>)} loops in ASLan++ are used to define processes
waiting for incoming messages (e.g., our database formalization).

\item Assignments are of the form \texttt{M:=m} where \texttt{M} is a variable and \texttt{m} is a constant.

\item There are various types in ASLan++, e.g.:
\texttt{agent} for entities involved in the communication, \texttt{message} for anything that can pass through the network, \texttt{text} for atomic messages (i.e., messages that do not contain concatenations, whereas \texttt{message} can contain a concatenation of \texttt{text}s), and the Boolean type \texttt{fact} used for predicates. ASLan++ supports also functions and sets whose types can be anything but \texttt{fact}.

\item The main ASLan++ keywords are: \texttt{Actor}, \texttt{?},
\texttt{iknows()}, which is a predicate that represents the intruder
knowledge (it stands for \emph{intruder knows}), and \texttt{fresh()}, which is a function that generates a fresh constant value that one can assign to a variable by \texttt{M:=fresh()}.

\item Finally, we describe here the only type of ASLan++ goal we have used in
our examples, i.e., \texttt{[](!iknows(M))} where \texttt{[]} is the
globally LTL (Linear Temporal Logic) operator, \texttt{!} expresses a
negation, \texttt{iknows} is as above and \texttt{M} (in our examples)
is a constant or the \texttt{tuple()} function. Informally, the goal states
that, in a safe ASLan++ specification, \texttt{M} will never be known by
the intruder, otherwise an attack is found.

\end{itemize}


An ASLan++ specification can be automatically translated (see~\cite{avantssar-tacas})
into a more low-level ASLan specification, which ultimately defines a
transition system $M=\langle \mathbf{S},\mathbf{I},\rightarrow \rangle$,
where $\mathbf{S}$ is the set of states, $\mathbf{I} \subseteq \mathbf{S}$ is the
set of initial states, and $\rightarrow \subseteq \mathbf{S}\times \mathbf{S}$ is
the (reflexive) transition relation. 
The structure of an ASLan specification is composed by six different sections:
signature of the predicates, types of variables and constants, initial state, Horn clauses,
transition rules of $\rightarrow$ and protocol goals. 
The content of the sections is intuitively described by their names. In particular, an initial state $I\in\mathbf{I}$ 
is composed by the concatenation of all the 
predicates that hold before running any rule
(e.g., the agent names and the intruder's own keys).
The transition relation $\rightarrow$ is defined as follows.
For all $S\in \mathbf{S}$, 
$S\rightarrow S'$ iff there exist a rule
such that  
$$
PP.NP\&PC\&NC\ifarrow[V]R
$$
(where $PP$ and $NP$ are
sets of positive and negative predicates, $PC$ and $NC$ conjunctions
of positive and negative atomic conditions) and a substitution 
$\gamma :\{v_1,\ldots,v_n\}\rightarrow T_\Sigma$
where $v_1,\ldots,v_n$ are the variables that occur in $PP$ and $PC$
such that:
(1) $PP\gamma\subseteq \lceil S\rceil^H$, where $\lceil S\rceil^H$
is the closure of $S$ with respect to the set of clauses $H$, (2)
$PC\gamma$ holds,
(3) $NP\gamma\gamma' \cap \lceil S\rceil^H = \emptyset$ for all substitutions $\gamma'$ such that $NP\gamma\gamma'$ is ground,
(4) $NC\gamma\gamma'$ holds for all substitutions $\gamma'$ such that $NC\gamma\gamma'$ is ground and (5)
$S'=(S\setminus PP\gamma)\cup R\gamma\gamma''$, where $\gamma''$ is any substitution such that for all $v \in V$, $v\gamma''$ does not occur in $S$.

We now define the translation of the ASLan++ constructs we have considered here.
Every ASLan++ entity is translated into a new \emph{state predicate} and
added to the section signature. This predicate is parametrized with respect
to a \emph{step label} (that uniquely identifies every instance) and
it mainly keeps track of the local state of an instance (current
values of whose variables)
and expresses the control flow of the entity by means of
step labels. 
As an example, if we have the ASLan++ entity 
\begin{lstlisting}[breaklines=true,numbers=none]
entity Snd(Actor, Rcv: agent){
 symbols
    Var: message;
}
\end{lstlisting}
then the predicate \texttt{stateSnd} is added to the section signature and, supposing an instantiation of the entity \texttt{new Snd(snd, rcv)}, the new predicate 
\texttt{state\_Snd(snd, iid, sl\_0, rcv, dummy\_message)}
is
used in transition rules to store all the informations of an entity,
where the ID \texttt{iid} identifies a particular instance, 
\texttt{sl\_0} is the step label, the parameters \mbox{\texttt{Actor},}
\texttt{Rcv} are replaced with constants \texttt{snd}
and \texttt{rcv}, respectively, and the message variable \texttt{Var}
is initially instantiated with \texttt{dummy\_message}.

Given that an ASLan++ is a hierarchy of entities, when an entity is
translated into ASLan, this hierarchy is preserved by a
\texttt{child(id\_1, id\_0)} predicate that states \texttt{id\_0}
is the parent entity of \texttt{id\_1} and both \texttt{id\_0}
and \texttt{id\_1} are entity IDs.

A variable assignment statement is translated into a transition rule 
inside the rules section.
As an example, if in the body of the entity \texttt{Snd} defined above
there is an assignment \texttt{Var := constant;} where 
\texttt{constant} is
of the same type of \texttt{Var}, then we obtain the following transition rule:
\begin{lstlisting}[numbers=none]
state_Snd(Actor,IID,sl,Rcv,Var)
=>
state_Snd(Actor,IID,succ(sl),Rcv,constant)
\end{lstlisting}
In the case of assignments to  
\texttt{fresh()},
the variable \texttt{Var} is assigned to a new variable.

In the case of a message exchange (sending or receiving statements),
the \texttt{iknows(message)} predicate 
is added to the right-hand side of the corresponding ASLan rule.
This states that the message \texttt{message} has been sent
over the network and \texttt{iknows} is used because, as is usual, the Dolev-Yao intruder is identified with the network itself.

The last point we describe is the translation of goals focusing 
only on the LTL goal we have used in our case studies.
Goals are translated into attack states containing the negation of the 
argument of the LTL operator:

\lstset{numbers=none}
\begin{lstlisting}
attack_state authorization :=
	iknows(M)
\end{lstlisting}

More information on ASLan, ASLan++ and the AVANTSSAR Platform can be found
in~\cite{avantssar-deliverable-2.3update,avantssar-tacas}.

\subsection{ASLan++ skeleton}
\label{sec:skeleton}
We now present a general ASLan++ specification that contains all the aspects
described in \autoref{sec:formalization}. We start by describing the first
part of the skeleton given in \autoref{lst:dbSymbols}, where we specify
agents, variables, constants, facts and uninterpreted functions used in the
overall specification (i.e., the entity \texttt{Environment}).

\lstset{numbers=left,xleftmargin=2em,framexleftmargin=1.5em}
\begin{lstlisting}[breaklines=true,caption={ASLan++ code of the symbols used in the skeleton of 
	the web app.},label={lst:dbSymbols}]
specification SpecificationSkeleton 
channel_model CCM
entity Environment {
symbols
 %entities involved in the communication
 webapp, database: agent;
 %DBMS
 nonpublic inDB(message): fact;
 nonpublic sanitizedQuery(message):message;
 nonpublic query(message): message;
 nonpublic tuple(message): message;
 nonpublic no_tuple: text; 
 %sql injection payload
 sqli: text; 
 nonpublic dashboard: text;
 nonpublic secureFolder: text;	
clauses 
 db_hc_ev(M): inDB(M.sqli);
\end{lstlisting}

Lines 1--3 begin the specification by stating a name (in this case,
\texttt{SpecificationSkeleton}) and the channel model used
(\texttt{CCM}), and by introducing the outermost entity\\
(\texttt{Environment}).
The symbols section of the environment begins in line 4, where we define the
constants representing agents involved in the specification (line 6). We only
need to represent the web app and the database but no client as we
defined in \autoref{sec:formalization}.
In lines 8--12, we define predicates, uninterpreted functions and constants as described in \autoref{sec:database}. 
They are defined as \texttt{nonpublic} to exclude them from the
initial knowledge of the DY intruder.
In line 14, we define the payload of \autoref{sec:dy} that the DY
intruder can use to perform \sqli{}. The constant \texttt{sqli} 
is public in order to add it to the DY intruder initial knowledge.
In lines 15--16, we define two constants representing two recurring
components of web apps: \texttt{dashboard}, which represents a user
administration page, \texttt{secureFolder}, which represents any secure
folder (or page) in a web app.
In line 17, the keyword \texttt{clauses} opens the Horn clauses section in
which the Horn clause in line 18 represents \texttt{db\_hc\_ev}
of \autoref{sec:database}.

Our skeleton considers two main subentities, defining the web app
(\autoref{lst:webAppEntity}) and the database
(\autoref{lst:dbEntity}).
These two entities are subentities of \texttt{Session} as showed in line 1 of
\autoref{lst:webAppEntity}, but we will come back to the session entity later
in this section.

\begin{lstlisting}[caption={ASLan++ code of the web app entity.},label={lst:webAppEntity}]
entity Session(Webapp, Database:agent){
 entity Webapp(Actor, Database:agent){
  symbols
   %all the symbols used in the body of this
   %entity in the body below
  body{ %write the behavior of the web app.
   %Every time a message M is sent to the 
   %database do not forget to add a nonce   
   %to avoid spurious replay attacks,
   %e.g., Actor *->* database: M.nonce;
   %the same for message received,
   %e.g., database *->* Actor: ?M.?Nonce }}
\end{lstlisting}

The \texttt{Webapp} entity has an empty \texttt{body} as it has
no fixed structure and depends on the particular app one is 
modeling; we add it with comments to help the modeling phase. As in any
other ASLan++ entity, the \texttt{symbols} section collects the type
definition of variables, constants and eventually functions or predicate used
in section \texttt{body}, where the behavior of the web app must be
specified following the formalization in \autoref{sec:webapp}.

The \texttt{Database} entity follows (almost verbatim) the pseudo-code in
\autoref{lst:databaseEntity}. The while loop in line 6 wraps the entire body
content. This defines that the database is listening for incoming
communications that match one of the guards of the two select-on in lines 9
and 14. The other two main differences are the usage of select-on instead of
an if statement and the introduction of nonces. The first is due to technical
reasons: when the ASLan++ specification is translated into a transition
system, the semantics of select-on with respect to an if statement saves one
or more transitions (thanks to the absence of the negative branch). The
introduction of nonces in any in/out-going message is useful to avoid
spurious man-in-the-middle attacks between the web app and the database in
accordance with \autoref{sec:webapp}.

\begin{lstlisting}[breaklines=true,caption={ASLan++ code of the database entity.},label={lst:dbEntity}]
entity Database(WebApp, Actor: agent){
symbols
 NonceWA,NonceDB: text;
 SQLquery: message;
body{
 while(true){
  select{ 
   on(WebApp *->* Actor: ?NonceWA.sanitizedQuery(?SQLquery)):{
    select{on(SQLquery = tuple(?)):{
     NonceDB := fresh();
     Actor *->* WebApp: no_tuple.NonceDB; } } }
   on(WebApp *->* Actor: 
   	?NonceWA.query(?SQLquery)):{
    select{
     on(inDB(SQLquery)):{
      NonceDB := fresh();
      Actor *->* WebApp: tuple(SQLquery).NonceDB; }
     on(!(inDB(SQLquery))):{
      NonceDB := fresh();
      Actor *->* WebApp: no_tuple.NonceDB; }}}}}}}
\end{lstlisting}

As we mentioned, both the web app and the database are
subentities of the session entity. In our skeleton, the only meaning of the
session entity is to instantiate its subentities with the proper constants
(defined in \autoref{lst:dbSymbols}) as showed in \autoref{lst:sessionBody}.

\begin{lstlisting}[breaklines=true,caption={ASLan++ code of the Session body.},label={lst:sessionBody}]
body{ new Webapp(webapp, database);  
      new Database(webapp, database); }
\end{lstlisting}

Goals (i.e., security properties that we want to check) 
are defined 
in \autoref{lst:aslanGoals}. They are a verbatim copy of the ones 
described in \autoref{sec:goals} but we use \texttt{iknows} instead of
\texttt{iknowledge}.

\begin{lstlisting}[caption={ASLan++ code for authentication and data extraction
	attacks.},label={lst:aslanGoals}]
goals
 authentication: [](!(iknows(dashboard))); 
 data_extraction: [](!(iknows(tuple(?)))); 
\end{lstlisting}

At the end of the specification, the \texttt{Environment} entity is
instantiated as in \autoref{lst:environmentBody}.

\lstset{numbers=left,xleftmargin=0em,framexleftmargin=0em}
\begin{lstlisting}[numbers=none,caption={ASLan++ code of the Environment body.},label={lst:environmentBody}]
body{ new Session(webapp, database); }
\end{lstlisting}

\section{ASlan++ Case Studies}
\label{sec:aslan-casestudies}
In this section, we give the ASLan++ code that implements the behavior of the
case studies presented in~\autoref{sec:casestudiesandresults} long with two
case studies coming from WebGoat. The ASLan++ of this section will fill the
empty spaces of the ASLan++ skeleton of~\autoref{sec:skeleton}.

\subsection{Authentication Bypass (WebGoat)}
\label{sec:aslanpp:webgoat}
We used Lesson \emph{Stage 1: String SQL injection} from WebGoat which implements a
common login scenario for authenticating users through username and password.
The web app receives credentials (\texttt{Employeeid} and \texttt{Password})
along with the \texttt{IP} address of the client (1).
A nonce \texttt{NonceWA} is generated (2) to ensure a fresh communication (avoiding spurious 
 replay
attacks) with the database and then a query with the provided credentials is sent to the
database entity (3).
The web app waits to receive an answer from the database (5-6) and, if some tuples
are received (6), 
the users list page (\texttt{usersList}) is returned
to the client along with the \texttt{tuple()} function (7).
If no tuples are generated (10), then the login page is shown (11).

\lstset{numbers=left,xleftmargin=2em,framexleftmargin=1.5em}
\begin{lstlisting}[breaklines=true,caption={ASLan++ code representing
the WebGoat Stage 1 scenario,  where
\texttt{Actor} refers to the web application.},label={lst:aslanpp:authby}]
? ->* Actor: ?IP.?Employeeid.?Password;
NonceWA:=fresh();
Actor *->* Database: NonceWA.query(Employeeid.Password);

select{
on(Database *->* Actor: tuple(?SQLquery).?NonceDB):{
 Actor ->* IP: usersList.tuple(SQLquery);
}

on(DB *->* Actor: no_tuple):{
 Actor ->* IP : login;
}}
\end{lstlisting}

As goal~(\autoref{lst:aslanpp:bbbgoal}), we check if there exists an execution in which
the intruder can obtain access to the \texttt{usersList} constant (representing 
the users list page) without knowing the correct credentials.

\begin{lstlisting}[breaklines=true,numbers=none,caption={Authentication goal.}
,label={lst:aslanpp:bbbgoal}]
[](!(iknows(usersList)));
\end{lstlisting}

\subsection{Data Extraction (WebGoat)}
Consider again the WebGoat lesson specified in~\autoref{sec:aslanpp:webgoat}.
We can use it again to show how to represent a data extraction attack. The
behavior of the web application does not change and is the same given
in~\autoref{lst:aslanpp:authby}.
As goal (\autoref{lst:aslanpp:dataleak}), we check if there is a way for the
intruder to gain knowledge of something which is function on \texttt{tuple()}.

The AAT obtained is identical to the one already given
in~\autoref{lst:aslanpp:authby} for authentication bypass. Not surprising,
whenever there is an injection point, the intruder can modify the behavior of
the SQL query by using different payloads to obtain different results. In our
model, we do not consider concrete payloads and then the attack execution is
the same. What allows us to distinguish between the two cases is the goal
itself, which gives us information about the intruder's intention.

\begin{lstlisting}[breaklines=true,numbers=none,caption={Data extraction goal}
,label={lst:aslanpp:dataleak}]
[](!(iknows(tuple(?)));
\end{lstlisting}

\subsection{Authentication bypass via Data Extraction}
The ASlan++ model representing the \joo{} case study in~\autoref{sec:authdata}
is given in~\autoref{lst:aslanpp:joomlamodel}.  A remote user, browsing the web
app, visits the Content History component (line 1).  The web app generates the
nonce \texttt{NonceWA} (2) to ensure a fresh communication (avoiding
spurious replay attacks) with the database and then sends a query with the user
supplied data (3).  The web application waits for a response from the database
(5).  If some tuples are generated (5), the web application sends back to the
client the history page along with the \texttt{tuple()} function (6).
The web application then has two possible ways of authenticating the user
(9--27): by using credentials or cookies. 
If username and password (12) are provided,
the web application applies a non-invertible hash function \texttt{hash()}
to the password, and queries the database to verify the credentials (13).
If the credentials are correct (14), the administration panel is
sent 
to the user (15).
In case of a cookie session (20), the user provides a cookie 
that the web application 
checks querying the database (21). If the cookie is valid (22), the administration panel is
sent back to the user (23).

\begin{lstlisting}[breaklines=true,caption={ASLan++ code representing the \joo{} scenario,  where \texttt{Actor} refers to the web application.},
label={lst:aslanpp:joomlamodel}]
? ->* Actor: ?IP.(com_contenthistory.history).?Listselect;
NonceWA := fresh();
Actor *->* Database: NonceWA.query((com_contenthistory.history).Listselect);

select{on(Database *->* Actor: tuple(?SQLquery).?NonceDB):{
    Actor ->* IP: viewHistory.tuple(SQLquery);
}}

% attempt to access a restricted area
select{
  % we are providing correct credentials
  on(IP ->* Actor: username.?Username.password.?Password):{
    Actor *->* Database: NonceWA.sanitizedQuery(Username.hash(Password));
    select{on(Database *->* Actor: no_tuple.?NonceDB ):{
          Actor ->* IP: adminPanel;
      }
    }
  }
  % we are providing a valid session
  on(IP ->* Actor: cookie.?Cookie):{
    Actor *->* Database: NonceWA.sanitizedQuery(Cookie);
    select{on(Database *->* Actor: no_tuple.?NonceDB):{
          Actor ->* IP: adminPanel;
      }
    }}

 }%end select
}%end body
\end{lstlisting}

As goal (\autoref{lst:aslanpp:joomlagoal}), we check if there exists an execution in
which the intruder can access the administration panel represented by the
constant \texttt{adminPanel}.

\begin{lstlisting}[breaklines=true,numbers=none,caption={Authentication bypass for the \joo{} scenario.}
,label={lst:aslanpp:joomlagoal}]
[](!(iknows(adminPanel)));
\end{lstlisting}

\subsubsection{\chained{}}
The ASlan++ model representing the \chained{} case study
in~\autoref{sec:authdata} is given in~\autoref{lst:aslanpp:chained}.
The client sends his personal credentials (\texttt{Username} and
\texttt{Password}) along with his IP address (\texttt{IP}) to the web application (1). The
web application generates the nonce \texttt{NonceWA} (2) to ensure a fresh
communication (avoiding spurious replay attacks) with the database and then creates a
query that it sends to the database (3). The web application waits for a
response of the database with a tuple (4--5). If tuples are generated from the
database (5), a dashboard page is returned to the client along with the tuple
(6). If no tuples are generated (7), the web application redirects the user to
the login page (8). At this point, the web application waits to receive
correct credentials (11) for the HTTP basic authentication that will allow the client to 
access the secure
folder \texttt{secureFolder} (12). Given that the credentials are the same as
the ones stored in the database, we can model them as a function of
\texttt{tuple()}. We recall, as discussed in \autoref{sec:database}, that
\texttt{tuple()} represents an abstraction of any data that can be extracted
from the database. This means that whenever a web application requires any
data that is in the domain of the database, we can write them as a function of
\texttt{tuple()}.

\begin{lstlisting}[breaklines=true,caption={ASLan++ code representing the \chained{}
				scenario, where \texttt{Actor} refers to the web application.},
label={lst:aslanpp:chained}]
? ->* Actor: ?IP.?Username.?Password;
NonceWA := fresh();
Actor *->* Database: NonceWA.query(Username.Password);
select{
 on(Database *->* Actor: tuple(?SQLquery).?NonceDB):{
    Actor ->* IP:dashboard.tuple(SQLquery);
 on(Database *->* Actor: no_tuple.?NonceDB){
    Actor ->* IP: login;
 }
}}
select{on(IP ->* Actor: tuple(?)):{
 Actor ->* IP: secureFolder;
}}
\end{lstlisting}

As a goal (\autoref{lst:aslanpp:chainedGoal}), we ask if there is an execution in which the intruder reaches \texttt{secureFolder}.
\begin{lstlisting}[breaklines=true,numbers=none,caption={Authorization goal}
,label={lst:aslanpp:chainedGoal}]
[](!(iknows(secureFolder)));
\end{lstlisting}

\subsection{Second-Order \sqli{} (SO)}
The ASlan++ model representing the case study in~\autoref{sec:secondorder}
is given in~\autoref{lst:aslanpp:secondOrder}.
The client sends a registration request together with his IP address (\texttt{IP}) and his 
personal
credentials (\texttt{Username} and \texttt{Password}) to the web application
(1). The web application generates a nonce \texttt{NonceWA} (2) to ensure a fresh communication with the 
database (to avoid spurious replay attacks) and sends a query containing the client's credentials 
along with the nonce to the database (3). The web application waits until it receives a 
response from the database containing the data resulting from the execution of the query 
\texttt{tuple(?SQLquery)} submitted by the web application (5).
The registration process is now completed and the web application sends back to the 
client the page \texttt{registered} (6). Here, the web application does not forward 
\texttt{tuple()} back to the client because the registration query is an \texttt{INSERT} (see 
\autoref{sec:database}).
The client asks for a page (8) and the web application uses previously submitted values of 
\texttt{Username} and \texttt{Password} to execute a new SQL query (9). Here is where the 
\so{} takes place; in fact, the variables embedded in the query in (9)
will trigger a \so{}.
The database executes the query and sends back the results to the web
application (11). Finally, (12), the web application sends to the client the requested page 
and the \texttt{tuple()} (the query submitted by the web application is a \texttt{SELECT}).

\begin{lstlisting}[breaklines=true,caption={ASLan++ code representing
a web application vulnerable to a \so{} attack where \texttt{Actor} refers to the web application.},
label={lst:aslanpp:secondOrder}]
? ->* Actor: registrationRequest.?IP.?Username.?Password;
NonceWA := fresh();
Actor *->* Database: NonceWA.query(Username.Password);

select{ on(Database *->* Actor: tuple(?SQLquery).?NonceDB):{
  Actor ->* IP: registered;
  
  IP *->* Actor: requestPage;
  Actor *->* Database: NonceWA.query(Username.Password);
  
  Database *->* Actor: tuple(?SQLquery).?NonceDB;
  Actor ->* IP: page.tuple(SQLquery);
}}
\end{lstlisting}

As goal (\autoref{lst:aslanpp:secondOrderGoal}), we ask if there is an execution in which the intruder can interact with the web application until he obtains data from the database, i.e.
\begin{lstlisting}[breaklines=true,numbers=none,caption={Database data leakage goal.},
,label={lst:aslanpp:secondOrderGoal}]
[](!(iknows(tuple(?))));
\end{lstlisting}

\subsection{Concretization phase}
\label{app:concretization}
We give the the output provided by \sqlfast{} with respect to the \joo{} case
study in~\autoref{lst:concretizationjoomla}.

\begin{lstlisting}[numbers=left,xleftmargin=2em,framexleftmargin=1.5em,breaklines=true,caption={
Concretization of the \sqli{} in the AAT \autoref{lst:joomlaaat}.},
label={lst:concretizationjoomla}]
Just a couple of questions.
What's the name of the web app in the ASLan++ specification? [DEFAULT webapplication, hit enter twice for default]

What's the name of the database in the ASLan++ specification? [DEFAULT database, hit enter twice for default]

Can you give me the URI of the web page under test corresponding to:
i -> WebApp: i.Username(4).sqli
0 - <?> ->* webapplication : com_contenthistory.history.sqli
http://target.com/joomla3.4.4/index.php?list[select]=?&view=history&option=com_contenthistory
----------------
0 - <?> ->* webapplication : com_contenthistory.history.sqli
http://target.com/joomla3.4.4/index.php?list[select]=?&view=history&option=com_contenthistory

Data extraction command
Are there any POST parameters (key=value)? [optional, press enter to skip]

sqlmap.py -u "https://157.27.244.25/joomla3.4.4/index.php?list[select]=?&view=history&option=com_contenthistory" -a
----------------
\end{lstlisting}

\sqlfast{} first asks to the modeler the names of the ASLan++ entities
representing the web application and the database (1-5). Then, \sqlfast{}
automatically detects the \sqli{} in the AAT and asks for the URL
corresponding to the message in the attack trace (lines 6-9). Once 
provided, \sqlfast{} asks whether there are any POST parameters for performing
the request to the give URL (15-16) and finally generates the appropriate 
command that need to be executed (17).

\end{EXT}
\end{document}